\documentclass[12pt]{amsart}
\usepackage{amssymb,amsmath}
\usepackage[dvips]{graphicx}
\usepackage{pslatex}
\usepackage{amsmath,amscd}
\usepackage{amsthm}
\usepackage{overpic}
\usepackage{color}

\usepackage{tikz}
\usetikzlibrary{shapes.multipart, arrows.meta, decorations.pathreplacing}
\usepackage{tabularx}
\usepackage{float}

\usepackage[hidelinks]{hyperref}

\theoremstyle{plain} \newtheorem{thm}{Theorem}[section]
 
\newtheorem{lemma}[thm]{Lemma}

\newtheorem*{namedtheorem}{\theoremname}
\newcommand{\theoremname}{testing}

\theoremstyle{definition} \newtheorem{defn}[thm]{Definition}

\newtheorem{ex}[thm]{Example} 
\theoremstyle{remark} \newtheorem{remark}[thm]{Remark}

\author{Christian Millichap}
\address{Department of Mathematics\\ 
Furman University\\ 
Greenville, SC 29613}
\email{christian.millichap@furman.edu}

\author{Yeeka Yau}
\address{Department of Mathematics \& Statistics\\ 
UNC Asheville\\ 
Asheville, NC 28804}
\email{yyau@unca.edu}

\author{Alyssa Pate}
\address{
Furman University\\ 
Greenville, SC 29613}
\email{pateal0@furman.edu}

\author{Morgan Carns}
\address{
Furman University\\ 
Greenville, SC 29613}
\email{carnmo5@furman.edu}

\usepackage{fancyhdr}

\fancypagestyle{plain}

\begin{document}

\title{Modifying twist algorithms for determining the key length of a Vigen\`{e}re cipher}

\begin{abstract}
In this article, we analyze and improve upon the twist-based algorithms introduced by Barr--Simoson and Park--Kim--Cho--Yum for determining the key length of a Vigen\`{e}re cipher. We provide an in-depth discussion on how the domain of the twist index affects the accuracy of these algorithms along with supporting experimental evidence. We also introduce a new twist-based algorithm, the twist$^{++}$ algorithm, and show this algorithm is more accurate than the twist$^{+}$ algorithm for a wide range of key lengths and text lengths.
\end{abstract}

\maketitle 
\let\thefootnote\relax\footnote{Address correspondence to Yeeka Yau \hyperlink{yyau@unca.edu}{yyau@unca.edu}, Department of Mathematics and Statistics, University of North Carolina Asheville, 1 University Heights, Asheville, NC 28804}


\section{Introduction}
\label{sec:intro}
The Vigen\`{e}re cipher is a well-studied topic in classical cryptology and frequently serves as the introductory example to polyalphabetic substitution ciphers in the literature. For this cipher, one fixes a keyword that determines a sequence of shifts to be used for encrypting individual letters, where the keyword repeats as needed to cover the length of the plaintext. The development of this cipher involved several people, starting with Leon Battista Alberti in the mid 1400s and leading up to Blaise de Vigen\'{e}re's \textit{Traict\'{e} des Chiffres} in 1585 \cite{Vi1586}. We refer the reader to Chapter 4 of \cite{Ka1996} for a detailed description of the history of the Vigen\`{e}re cipher and Section 3.5 of \cite{Do2018} for a brief description. The Vigen\'{e}re cipher has also seen a wide range of applications in history and popular culture, including communications between Confederacy officers during the American Civil War \cite{Bo2006}, \cite{Bo2016} and the CIA's Kryptos sculpture \cite{BaLiMo2016}.

The first step in cryptanalyzing the Vigen\`{e}re cipher, or any polyalphabetic substitution cipher, is finding the key length. The two most well-studied tools for this process are the Babbage--Kasiski test \cite{Ka1863} from the 1860s and the index of coincidence (IC) \cite{Fr1920}, which was introduced by William Friedman in 1920. We refer the reader to \cite[Chapter 2]{Ba2021} for background on these topics. While both of these tools can be useful for finding the key length, in many cases, they run into accuracy and implementation issues. For instance, if the Babbage--Kasiski test predicts a key length of $m$, then all divisors of $m$ are also reasonable conjectures for the key length. This test also won't be helpful if very few (or no) repeated tri-grams are found in the ciphertext. At the same time, the IC decreases in accuracy as key length increases, which can  be seen experimentally in Figure 2 of \cite{Ma1988} and Figure 3 of \cite{PaKiChYu2020}. In addition, the key length estimate determined by the IC is dependent on your keyword containing $k$ distinct letters. If many letters are repeated in the keyword, then the IC will sometimes underestimate $k$. Through the 20th century, a few other methods have been proposed for finding the key length (see \cite{Ma1988} for one example), though all those methods seem to only be accurate for particular ranges of key length and message length. These challenges motivate the development of more robust key length attacks.

More recently, the twist algorithm \cite{BaSi2015} was introduced by Barr and Simoson in 2015 as a new key length attack, which finds the potential key length $m$ that maximizes a ``twist index.'' This index measures the difference between the number of frequent and infrequent letters that are conjectured to be encrypted with the same shift based on a predicted key length of $m$; see Section \ref{subsec:Twist} for precise details. An improved version of this algorithm called the twist$^{+}$ algorithm was then introduced by Park--Kim--Cho--Yum in 2020. An important observation made by Park--Kim--Cho--Yum which led to the development of the twist$^{+}$ algorithm is that if you conjecture a key length of $m$ for a fixed ciphertext, then the twist index is nondecreasing as a function of multiples of $m$. This fact is only verified experimentally in their paper, but plays a crucial role in highlighting a major weakness of the twist algorithm: it is not clear how to distinguish the  key length from its multiples with the twist algorithm. This fact is essential for motivating both the twist$^{+}$ algorithm and our own modification to the twist algorithm. We give a proof that this property holds under certain conditions in Lemma \ref{lem:lambda m bigger than m} to provide more conclusive evidence of the twist algorithm's major flaw. 

While Park--Kim--Cho--Yum provide experimental evidence (see Figure 3 of \cite{PaKiChYu2020}) to support the accuracy of the twist$^{+}$ algorithm, there are  still  conditions where this algorithm does not perform well. As noted in the conclusion of their paper, the twist$^{+}$ algorithm is less effective for short keys.  In addition, the domain of possible key lengths to check can play a significant role in the effectiveness of both the twist and twist$^{+}$ algorithm, with a larger domain of values resulting in a decrease in accuracy. We examine this subtlety in detail in this paper, providing experimental evidence to back up our claims. In both types of scenarios, the twist$^{+}$ algorithm usually conjectures a multiple of the key length. To alleviate this weakness, we introduce our own (simple) modification of the twist algorithm, which we call the twist$^{++}$ algorithm. We provide experimental evidence showing that the twist$^{++}$ algorithm outperforms the twist$^{+}$ algorithm for a wide range of text lengths and key lengths, while being far less susceptible to decreases in accuracy as the domain of potential key lengths increases. We also provide an analysis of how the twist$^{++}$ algorithm behaves when it conjectures an incorrect key length. While this scenario is rare, it can occur under  particular conditions: when we have a relatively short text, and a relatively small  composite key length $k$, where $k=2m$ for some $m \in \mathbb{N}$. Under such conditions, the twist$^{++}$ algorithm sometimes predicts a key length of $m$, and we carefully explain why this is the case.


\section{The Twist Algorithm and the Twist$^{+}$ Algorithm}
\label{subsec:Twist}

In this section, we  review the twist algorithm introduced by Barr--Simoson \cite{BaSi2015} and the twist$^{+}$ algorithm introduced by Park--Kim--Cho--Yum \cite{PaKiChYu2020}. We also provide an updated analysis of these algorithms. Suppose we are given a text $\mathcal{M}$ of length $N$, with  text characters labeled in order as $x_1, \ldots, x_N$. In all relevant cases, $\mathcal{M}$ will be a ciphertext that was encrypted using the Vigen\`{e}re cipher. Under these conditions, let $m$ be a conjectured key length. We can partition our ciphertext into $m$ cosets, $\{\mathcal{M}_j\}_{j=1}^{m}$, where $\mathcal{M}_j = \{ x_i \hspace{0.02in} :  \hspace{0.02in} i \equiv j  \mod m\}.$ For each coset, we calculate the relative frequencies of each letter and then order these frequencies from smallest to largest. We let $C_j = <c_{1,j}, \hspace{0.02in} c_{2,j}, \hspace{0.02in} \ldots, \hspace{0.02in}  c_{26,j} >$ represent the vector of ordered relative frequencies for coset $M_j$, called the \textbf{sample signature} for $M_j$. 

If the conjectured key length $m$ is in fact the actual key length, $k$, then we should expect each  $C_j$ to approximate the behavior of ordered relative frequencies in the underlying language (English for all of our cases). If $m \neq \lambda k$ for $\lambda \in \mathbb{N}^{+}$, then each $C_j$ should more closely approximate the frequencies of a random text, which are more uniform than in English. The twist algorithm attempts to capture this distinction by analyzing the difference between the relative frequencies of the 13 most frequent letters and the 13 least frequent letters in each coset for a conjectured key length $m$. This motivates the \textbf{twist} of a sample signature: 

$$\diamondsuit C_j = \sum_{i=14}^{26} c_{i,j} - \sum_{i=1}^{13} c_{i,j}.$$

If a coset behaves  like a random text, then $\diamondsuit C_j$ should be small, while if a coset only contains letters from the same alphabet, then we should expect each $\diamondsuit C_j$ to be relatively large. We  now can define the twist algorithm in terms of maximizing the sum of the twists of sample signatures among possible key lengths, as introduced by Barr--Simoson \cite{BaSi2015}.

\begin{defn} \label{def:twist}
Let $\mathcal{M}$ be a text of length $N$. The \textbf{twist algorithm} finds $m \in \mathbb{N}^{+}$ that maximizes the \textbf{twist index}
$$T(\mathcal{M}, m) = \Big(\frac{100}{m}\sum_{j=1}^{m} \diamondsuit C_j\Big).$$
\end{defn}

\begin{remark}\label{remark0} In \cite{BaSi2015}, Barr--Simoson defined the twist index by rounding $T(\mathcal{M},m)$  to the nearest integer. We dropped the round function from the definition  since the version of $T(\mathcal{M},m)$ in Definition \ref{def:twist}  gives a more precise answer for which $m \in \mathbb{N}^{+}$ maximizes the twist index and it simplifies notation. 
\end{remark}

\begin{remark}
\label{remark1}
In both \cite{BaSi2015} and \cite{PaKiChYu2020}, there are no upper bound restrictions formally stated for $m$ in the definition of the twist algorithm, though only $1 \leq m \leq Q$, with $Q \in \{ 9,10\}$, are shared in the examples and data in these papers. For our purposes, we will always assume $m \in \{ 1, \ldots, q\}$, where $N = 12q + r$ for quotient $q$ and remainder $r$. When examining $T(\mathcal{M}, m)$ for $m > q$, we always have $T(\mathcal{M}, m) = 100$ since under these conditions each coset $\mathcal{M}_j$ contains at most thirteen distinct letters, and so, $\diamondsuit C_j$ is equal to the sum of the relative frequencies of all the letters in $\mathcal{M}_j$. This also highlights that the twist algorithm is not useful when the ratio $\frac{k}{N}$ is sufficiently large.  For experimental purposes, we might restrict $m$ to  smaller values, though we will always formally state what $m$-values are in the domain of $T(\mathcal{M},m)$ for the given analysis. 
\end{remark}

Challenges with the twist algorithm occur when $m = \lambda k$, for $\lambda \in \mathbb{N}^{+}$ and $\lambda \geq 2$. First, experimental evidence from \cite{PaKiChYu2020} shows that $T(\mathcal{M}, k) \leq T(\mathcal{M}, \lambda k)$ for any $k \in \mathbb{N}^{+}$, and we formally verify this result  under certain conditions in Lemma \ref{lem:lambda m bigger than m}. This is not surprising since each coset $\mathcal{M}_{j}$ relative to key length conjecture $m=k$ gets partitioned into $\lambda$ subsets, $\{\widetilde{\mathcal{M}}_{lm+j}\}_{l=0}^{\lambda -1}$ to form the cosets relative to key length conjecture $m = \lambda k$. Thus, if $k$ is the actual key length, then the letters in each coset $\widetilde{\mathcal{M}}_{lm+j}$ should still mirror the relative frequencies of the underlying language, and so, $T(\mathcal{M}, \lambda k)$ should still be a large value. As a result of Lemma \ref{lem:lambda m bigger than m} (and the experimentally verified behavior for the general case), we expect the twist algorithm to predict the largest possible multiple of the key length, assuming nontrivial multiples of the key length are part of the domain of $T(\mathcal{M}, m)$. While this is still useful information, Park--Kim--Cho--Yum noticed a modification could be made to increase accuracy for predicting the actual key length rather than a multiple. As noted in their paper, ``At $m= \lambda k$ for $\lambda \in \mathbb{N}^{+}$, the twist index $T(\mathcal{M}, m)$ usually has its local maximum. In addition, $T(\mathcal{M}, m)$ increases rapidly at these points. Therefore, we propose to find the key length by searching for the first point where the twist index increases rapidly.'' To formalize this idea, Park--Kim--Cho--Yum proposed that one should subtract the average of the previous $m-1$ twist indices from $T(\mathcal{M}, m)$ in order to more likely distinguish a global maxima from a local maxima. More precisely, we have the following definition.

\begin{defn}\label{defnTwistPlus}
Let $\mathcal{M}$ be a text of length $N$. The \textbf{twist$^{+}$ algorithm} finds $m \in \mathbb{N}^{+}$ that maximizes the \textbf{twist$^{+}$ index}
$$T^{+}(\mathcal{M}, m) = T(\mathcal{M},m) - \frac{1}{m-1} \sum_{\mu=1}^{m-1}T(\mathcal{M}, \mu),$$
where $m \in S \subseteq \{1, \ldots, q\}$ and $N = 12q + r$ for quotient $q$ and remainder $r$.
\end{defn}

Park--Kim--Cho-Yum verified the effectiveness of the twist$^{+}$ algorithm by analyzing plaintexts of length $N \in \{ 200, 300, 400, 500\}$ which were encrypted with keys of length $k \in \{2, \ldots, 10\}$, and comparing accuracy between the index of coincidence, the twist algorithm, and the twist$^{+}$ algorithm; see Figure 3 from \cite{PaKiChYu2020}. This data shows that the twist$^{+}$ algorithm is a significant improvement over the index of coincidence and the twist algorithm for the given parameters. However,  it is unclear what the exact domain of $m$-values were considered for this analysis, though it must have included $m$-values at least as large as $10$ since actual keys of length up to $10$ were used for encryption. If the only $m$-values considered were $m \in \{1,\ldots, 10\}$, then this could lead to a different conclusions about the accuracy of the twist$^{+}$ algorithm than if a larger set of $m$-values was considered, especially for the accuracy of larger key lengths. See Figure \ref{fig:twistdoublecomparison} and Figure \ref{fig:twistdoublecomparison_m20} for experimental evidence that support this claim. In addition, see Example \ref{TwistEx1} for a specific scenario where the twist$^{+}$ algorithm maximizes at index $m = 16$ (or possibly higher), but the actual key length is $k=4$. 

As a function of $m$, it is somewhat unclear how to determine an appropriate domain for $T^{+}(\mathcal{M},m)$. The largest domain one should consider is $m \in \{1,\ldots, q\}$, though a specific subset of these values could be used, as designated in Definition \ref{defnTwistPlus}. Smaller upper bounds might be more reasonable for experimental purposes, though this probably depends on what one considers ``fair'' test parameters. If we want to build a test where it is assumed that a key length $k$ can only be $k \in \{1,\ldots, Q\}$, for some fixed $Q$ known ahead of time, then there is no need to consider index values larger then $Q$. However, if it is unclear what the largest key length value might be, then it could make sense to include $m$-values as large as experimentally feasible, up to $q$.

To end this section, we provide a proof that $T(\mathcal{M}, m) \leq T(\mathcal{M}, \lambda m)$ under specific conditions, along with some remarks on the general case. This behavior of $T(\mathcal{M},m)$ was a crucial insight that motivated modifying the twist algorithm and was only verified experimental in \cite{PaKiChYu2020}.

\begin{lemma} \label{lem:lambda m bigger than m}
    Let $\mathcal{M}$ be a  text, and $m, \lambda \in \mathbb{N}^+$. If $\lambda$ divides  the size of each $m$-coset, then $T(\mathcal{M}, m) \le T(\mathcal{M}, \lambda m)$.
\end{lemma}

\begin{proof} 
    Fix $m, \lambda \in \mathbb{N}^+$. Then $m$ gives a partition of $\mathcal{M}$ into the set of $m$-cosets $\{\mathcal{M}_j\}_{j=1}^{m}$ with $C_j$ denoting the corresponding sample signature for $\mathcal{M}_j$. Likewise, $\lambda m$ gives a partition of $\mathcal{M}$ into the set of $\lambda m$-cosets $\{\mathcal{\widetilde{M}}_j\}_{j=1}^{\lambda m}$ with  $\widetilde{C}_j$ denoting the corresponding sample signature of $\mathcal{\widetilde{M}}_{j}$. Note that each $m$-coset $\mathcal{M}_j$ is a disjoint union of $\lambda m$-cosets. In particular, for an $m$-coset $\mathcal{M}_j$ we can write:
    \begin{equation} \label{cj_subcosets}
        \mathcal{M}_j = \mathcal{\widetilde{M}}_j \sqcup \mathcal{\widetilde{M}}_{m + j} \sqcup \ldots \sqcup \mathcal{\widetilde{M}}_{(\lambda-1) m + j}
    \end{equation}
    Thus, to prove  $T(\mathcal{M}, m) \le T(\mathcal{M}, \lambda m)$, it suffices to show that for all $j=1,\ldots, m$, we have 
    
    \begin{equation} \label{lem:what_we_need_to_show}
        \diamondsuit C_j \le \frac{1}{\lambda} \big( \sum_{l=0}^{\lambda -1} \diamondsuit \widetilde{C}_{lm + j} \big)
    \end{equation}

    Before proceeding, we set up some notation and rephrase Inequality \ref{lem:what_we_need_to_show}. Let $C:= C_j$ for some fixed $j$ and let $\widetilde{C}_{l}$ for $0 \le l \le \lambda - 1$ be the sample signatures corresponding to the (sub)-coset $\mathcal{\widetilde{M}}_{lm +j}$ from the right-hand side of Equation \ref{cj_subcosets}. Let $|\mathcal{M}_j| = N$ and  $C = < c_1, \ldots, c_{26}>$ be the vector of ordered relative frequencies to coset $\mathcal{M}_j$. \par
    Define 
    $v := <Nc_1, \ldots, Nc_{26}> = <f_{1}, \ldots, f_{26}>$ i.e, $v$ denotes the ordered vector of raw frequencies corresponding to coset $C$. Let $v_L  = <f_{1}, \ldots, f_{13}>$ and $v_R = < f_{14}, \ldots, f_{26}>$. Similarly, let $\tilde{v}^{l} = < \tilde{f}_{1}^{l}, \ldots, \tilde{f}_{26}^{l}>$ denote the ordered vector of raw frequencies corresponding to (sub)-coset $\mathcal{\widetilde{M}}_{lm +j}$, and we use $\tilde{v}_{L}^{l}$ and $\tilde{v}_{R}^{l}$ to denote the 13-dimensional vectors corresponding to the first 13 entries of $\tilde{v}^{l}$ and the second 13 entries of $\tilde{v}^{l}$, respectively. Note that,  $\diamondsuit C = \frac{1}{N}\bigg(\displaystyle\sum_{i=14}^{26} f_i - \displaystyle\sum_{i=1}^{13} f_i\bigg)$.
    
Since we are assuming that $\lambda$ divides the size of each $m$-coset, we have that $\lambda$ divides $N$.  Then since each $|\widetilde{C}_l| = \frac{N}{\lambda}$, we have  $$\bigg( \displaystyle\sum_{k=0}^{\lambda -1} \diamondsuit \widetilde{C}_{km + j}\bigg) = \frac{\lambda}{N} \bigg[\sum_{k=0}^{\lambda-1} \bigg(\sum_{i=14}^{26} \widetilde{f}_{i}^{l} - \sum_{i=1}^{13}\widetilde{f}_{i}^{l}\bigg)\bigg].$$    
    Thus, we can can rephrase Inequality \ref{lem:what_we_need_to_show}  as

    \begin{equation} \label{inequalityrestated}
   \sum_{i=14}^{26} f_i - \sum_{i=1}^{13} f_i \leq \sum_{l=0}^{\lambda-1} \bigg(\sum_{i=14}^{26} \widetilde{f}_{i}^{l} - \sum_{i=1}^{13}\widetilde{f}_{i}^{l}  \bigg).
  \end{equation}

    In order to prove Inequality \ref{inequalityrestated} we need to make a connection between the raw frequencies $f_i$ for coset $\mathcal{M}_j$ and the raw frequencies $\widetilde{f}_{i}^{l}$ of each (sub)-coset $\widetilde{\mathcal{M}}_{lm+j}$, for $l = 0, \ldots, \lambda -1$. A key observation is that the ordered raw frequency vectors $\tilde{v}^{l}$ can be constructed as follows: given $v =  <f_{1}, \ldots, f_{26}>$, let $f_{i}^{alph}$ designate  the  letter that is counted by $f_{i}$. Then define $f_{i}^{l}$ as the number of occurrences of letter $f_{i}^{alph}$ in $\widetilde{\mathcal{M}}_{lm+j}$.
   Then to construct $\tilde{v}^{l} = < \tilde{f}_{1}^{l}, \ldots, \tilde{f}_{26}^{l}>$, we sort the entries of vector $<f _{1}^{l}, \ldots, f_{26}^{l}>$
    from smallest to largest. Note that, by construction $f_i = \sum_{l = 0}^{\lambda -1} f_{i}^{l}$.

We now claim that $\sum_{i = 14}^{26} f_i  \le    \sum_{l = 0}^{\lambda -1} \sum_{i = 14}^{26} \tilde{f}_{i}^{l}$
Since $\sum_{i = 14}^{26} f_i  =    \sum_{l = 0}^{\lambda -1} \sum_{i = 14}^{26} f_{i}^{l}$, we just need to show $\sum_{i=14}^{26} f_{i}^{l} \le \sum_{i=14}^{26} \tilde{f}_{i}^{l}$ for each $l= 0 , \ldots, \lambda -1$.   For this, we can compare the vectors $v_{R}^{l}$ and $\tilde{v}_{R}^{l}$. Note that, $\sum_{i=14}^{26} f_{i}^{l}$ is  the sum of the entries of $v_{R}^{l}$ and $\sum_{i=14}^{26} \tilde{f}_{i}^{l}$ is the sum of the entries of $\tilde{v}_{R}^{l}$. If $v_{R}^{l}$ is a permutation of  $\tilde{v}_{R}^{l}$, then $\sum_{i=14}^{26} f_{i}^{l} = \sum_{i=14}^{26} \tilde{f}_{i}^{l}$. If $v_{R}^{l}$ is not a permutation of  $\tilde{v}_{R}^{l}$, then some values in $v_{L}^{l}$ are larger than some values in $v_{R}^{l}$. Thus, when we sort $v^{l}$ to put its entries in increasing order to create $\tilde{v}^{l}$, we could only increase the sum of the entries in the right-half of these respective vectors, i.e., $\sum_{i=14}^{26} f_{i}^{l} < \sum_{i=14}^{26} \tilde{f}_{i}^{l}$. This proves the claim and a similar process shows that $\sum_{i = 1}^{13} f_i  \ge  \sum_{l = 0}^{\lambda -1} \sum_{i = 1}^{13} \tilde{f}_{i}^{l}$, which then implies Inequality \ref{inequalityrestated}, giving the desired result. 
    \end{proof}

    \begin{remark}
A general proof of this fact could possibly proceed in the same manner, until we rely on our assumption that $\lambda$ divides $N$. If $\lambda$ does not divide $N$, then suppose $N = \lambda \alpha +r$, where $\alpha = \lfloor N / \lambda \rfloor$ and $1 \le r < \lambda$. We would want to show Inequality \ref{lem:what_we_need_to_show} holds under these conditions. In this case, the $\lambda$-cosets can be partitioned into two sets: $\{\mathcal{\widetilde{M}}_j\}_{j=0}^{r-1}$, each with $\alpha+1$ elements and $\{\mathcal{\widetilde{M}}_j\}_{j=r}^{\lambda  -1}$, each with $\alpha$ elements. However, the fact that the number of elements in these cosets varies creates challenges in moving to some version of Inequality \ref{inequalityrestated}. That being said, we experimentally verified that $T(\mathcal{M}, m) \le T(\mathcal{M}, \lambda m)$ still holds under this scenario for approximately 100,000 cases.
    \end{remark}



\section{The Twist$^{++}$ Algorithm}
\label{sec:twistdoubleplus}

While the twist$^{+}$ algorithm is frequently an effective tool for determining the key length of a Vigen\`{e}re cipher, there are  conditions where this algorithm does not perform well. Accuracy issues for the twist$^{+}$ algorithm can occur when we have short key lengths. Experimentally, the top left graph in Figure 3 of \cite{PaKiChYu2020} shows lower success rates for the twist$^{+}$ algorithm when $k \in \{2,3,4,5\}$ and $N = 200$.  In addition, as mentioned in Section \ref{subsec:Twist}, increasing the size of your domain of $m$-values for $T^{+}(\mathcal{M},m)$ can decrease the accuracy of the twist$^{+}$ algorithm. Since in many scenarios it will be unclear when you can restrict to a small set of $m$-values, this is another concern. In both of these scenarios, the twist$^{+}$ algorithm frequently  runs into issues similar to the twist algorithm had with predicting a nontrivial multiple of the actual key length.


This raises the question: Can we further modify the twist algorithm to alleviate this issue? Here, we propose one adjustment to the twist algorithm where we  search for the value where the twist index increases rapidly and then decreases rapidly relative to a local neighborhood rather than compare the current twist index to the average of all of the  previous twist indices. 

\begin{defn}
Let $\mathcal{M}$ be a text of length $N$. The \textbf{twist$^{++}$ algorithm} finds $m \in \mathbb{N}^{+}$ that maximizes the \textbf{twist$^{++}$ index}
$$T^{++}(\mathcal{M}, m) = T(\mathcal{M},m) - \frac{1}{2} \Big(T(\mathcal{M}, m-1) + T(\mathcal{M}, m+1)\Big),$$
where $m \in S \subseteq \{2, \ldots, q\}$ and $N = 12q + r$ for quotient $q$ and remainder $r$.
\end{defn}

Collectively, we refer to the twist algorithm, twist$^{+}$ algorithm, and twist$^{++}$ algorithm as the \textbf{twist-based algorithms.}

To directly compare the twist$^{++}$ algorithm to the twist$^+$ algorithm we replicated the parameters of the data set for which the twist$^+$ algorithm was analyzed against in \cite{PaKiChYu2020}. The data set included 100 texts of lengths $N \in \{ 200, 300, 400, 500 \}$ each encrypted with $40$ keys for the lengths $2 \le k \le 10$. Figure \ref{fig:twistdoublecomparison} demonstrates the performance of the twist$^+$ algorithm versus the twist$^{++}$ algorithm when we consider $2 \le m \le 15$. The twist$^{++}$ algorithm performs significantly better overall, especially when $N \in \{ 200, 300, 400 \}$ for $ 2\le k \le 7$. \par

Figure \ref{fig:twistdoublecomparison_m20} shows that the difference in accuracy increases between  the twist$^+$ algorithm and twist$^{++}$ algorithm when the values $2 \le m \le 20 $ are considered. As discussed in Section \ref{subsec:Twist}, the twist$^+$ algorithm is more likely to predict a multiple of the key length as the domain of $m$-values increases, which is verified experimentally here. However, The twist$^{++}$ algorithm still maintains a high level of accuracy  under these parameters.

For $2 \leq m \leq 20$, the \textit{only} parameters where the twist$^{+}$ algorithm outperforms the twist$^{++}$ algorithm is when $N=200$ and $k \in \{8,9,10\}$. Under these conditions, the twist$^{++}$ algorithm frequently predicts the largest nontrivial divisor of the key length. For instance, out of $4000$ texts of length $N=200$ encrypted with a key of length $k=8$, the twist$^{++}$ algorithm  incorrectly predicted the key length   $721$ times. Out of these $721$ texts, the twist$^{++}$ algorithm predicted a key length of $m=4$ (i.e., the largest nontrivial divisor of $k=8$), approximately $91.68\%$ of the time. See Table \ref{Table:Twistplusplusfails} for data on the cases where $k=8,9,10$. It is interesting to see that that when these twist-based algorithms fail, their predictions usually go in different directions: the twist$^{+}$ algorithm frequently predicts a nontrivial multiple of the actual key length, while the twist$^{++}$ algorithm frequently predicts the largest nontrivial divisor of the actual key length. 

For the twist$^{++}$ algorithm, let us discuss why we expect incorrect predictions to frequently be the largest nontrivial divisor of the actual key length $k$. First, let us  consider how $T^{++}(\mathcal{M},m)$ behaves for $m<k$. If $m$ is a divisor of $k$, then the $k$-cosets determined by the actual key can be partitioned into groups of $\frac{k}{m}$ sets, each making up one of the $m$-cosets determined by the predicted key length of $m$. When $m$ is large as possible, we are minimizing the number of (different) $k$-cosets we are combining into an $m$-coset. This should still result in a relatively large $T(\mathcal{M}, m)$-value since we have only combined a few alphabets together. At the same time, if $m$ is the largest divisor of $k$, then  $m+1$  will not be a divisor of $k$, and $m-1$ will sometimes not be a divisor of $k$; see Remark \ref{remark:divisors}. As a result, $T(\mathcal{M}, m+1)$  should usually be much smaller than $T(\mathcal{M}, m)$ since the $(m+1)$-cosets  will frequently behave closer to a random sample rather than combining samples from a few alphabets. The same behavior should be expected for $T(\mathcal{M}, m-1)$ when $m-1$ does not divide $k$. Thus, $T^{++}(\mathcal{M}, m)$ should be relatively large at such $m$-values. Now suppose $m$ does not divide $k$, then the $m$-cosets should behave more like random samples, and so, $T(\mathcal{M}, m)$ should be relatively small, and likewise $T^{++}(\mathcal{M},m)$ will too. Thus, for $m<k$, then we should expect $T^{++}(\mathcal{M},m)$ to likely have a global maxima at the largest divisor of $k$. See Example~\ref{ex:twist++explain} for a short example analyzing this coset behavior. 

\begin{ex}
\label{ex:twist++explain}
Consider the following part of a cipher text with $k = 6$.
\begin{center}
    DYECYLYLWCAQGGNGHBAEZUHQ \\
\end{center}
We color the letters with the same equivalence class mod 6, with the colors blue, black, cyan, red, green, and purple, respectively.
\begin{center}
    \textcolor{blue}{D}
    Y
    \textcolor{cyan}{E}
    \textcolor{red}{C}
    \textcolor{green}{Y}
    \textcolor{purple}{L}
    \textcolor{blue}{Y}
    L
    \textcolor{cyan}{W}
    \textcolor{red}{C}
    \textcolor{green}{A}
    \textcolor{purple}{Q}
    \textcolor{blue}{G}
    G
    \textcolor{cyan}{N}
    \textcolor{red}{G}
    \textcolor{green}{H}
    \textcolor{purple}{B}
    \textcolor{blue}{A}
    E
    \textcolor{cyan}{Z}
    \textcolor{red}{U}
    \textcolor{green}{H}
    \textcolor{purple}{Q}
\end{center}
When $m = 3$ we see that the first coset $C_{1}^{3}$ only contains letters from two alphabets (letters in blue and red), whereas when $m = 2$ or $m = 4$, the first cosets $C_{1}^{2}$ and $C_{1}^{4}$ both contain letters coming from three alphabets (letters in blue, green and cyan). Furthermore, if we consider $m \in \mathbb{N}$ such that $gcd(m,k)=1$, then any $m$-coset $C^{m}_{i}$ will contain letters from a wide range of alphabets, and given a long enough message, $C^{m}_{i}$ will contain letters from each of the $k$-alphabets. For instance, if we consider $C^{5}_{1}$ for this example, then this coset contains letters from all five alphabets (all colors represented).

\end{ex}

Now, we will  consider  how $T^{++}(\mathcal{M},m)$ behaves for $m>k$. First, as $m$ gets sufficiently large, we quickly see less variation in $T(\mathcal{M},m)$-values, and so, $T^{++}(\mathcal{M},m)$-values will be relatively small. Thus, we will rarely expect the twist$^{++}$ algorithm to predict an  $m$-value significantly larger than $k$, which will frequently make (nontrivial) multiples of $k$ unlikely candidates to maximize $T^{++}(\mathcal{M},m)$. At the same time, if $m$ is not a multiple of $k$, then the $m$-cosets behave closer to random samples and we are unlikely to have large values for both $T(\mathcal{M},m)$ and $T^{++}(\mathcal{M},m)$, as discussed in the previous paragraph. Thus, we should expect  $T^{++}(\mathcal{M},m)$-values to usually be relatively small, whenever $m>k$. 

Based on the previous  two paragraphs, we can see that if $m \neq k$, then $T^{++}(\mathcal{M},m)$ should frequently have a global maxima when $m$ is the largest nontrivial divisor of $k$. As a result, if the twist$^{++}$ algorithm makes an incorrect prediction, then it should frequently predict the largest nontrivial divisor of the actual key length. This also implies that we rarely expect the twist$^{++}$ algorithm to make incorrect predictions for prime key lengths. 

\begin{remark}\label{remark:divisors}
Suppose $m$ is the largest nontrivial divisor of $k \in \mathbb{N}$ with $k>1$ and $k$ composite. We claim that if $m-1$ is composite, then $m-1$ can not divide $k$. Since $m$ is the largest nontrivial divisor of $k$, we must have $k = m \cdot p$, for some prime $p$ (where $p$ is the smallest prime in the prime factorization of $k$). Suppose $m-1$ divides $k$. Since $gcd(m-1,m)=1$ and $m-1$ divides $m \cdot p$, we have that $m-1$ divides $p$. Since $p$ is prime, we can conclude that $m-1=p$, proving the claim. 
Thus, when we say ``$m-1$ will sometimes not be a divisor of $k$'' we know this will occur at least whenever $m-1$ is composite. 
\end{remark}

\begin{remark}\label{remark:sufficientlylarge}
In the discussion above, we are intentionally making ``as $m$ gets sufficiently large'' imprecise, though this condition is dependent on both $m$ and $N$ since we need $\frac{N}{m}$ to become sufficiently small for $T(\mathcal{M},m)$-values to stay within a sufficiently small neighborhood of $100$.  
\end{remark}


\begin{figure}[H]
    \centering
    \includegraphics[scale=0.33]{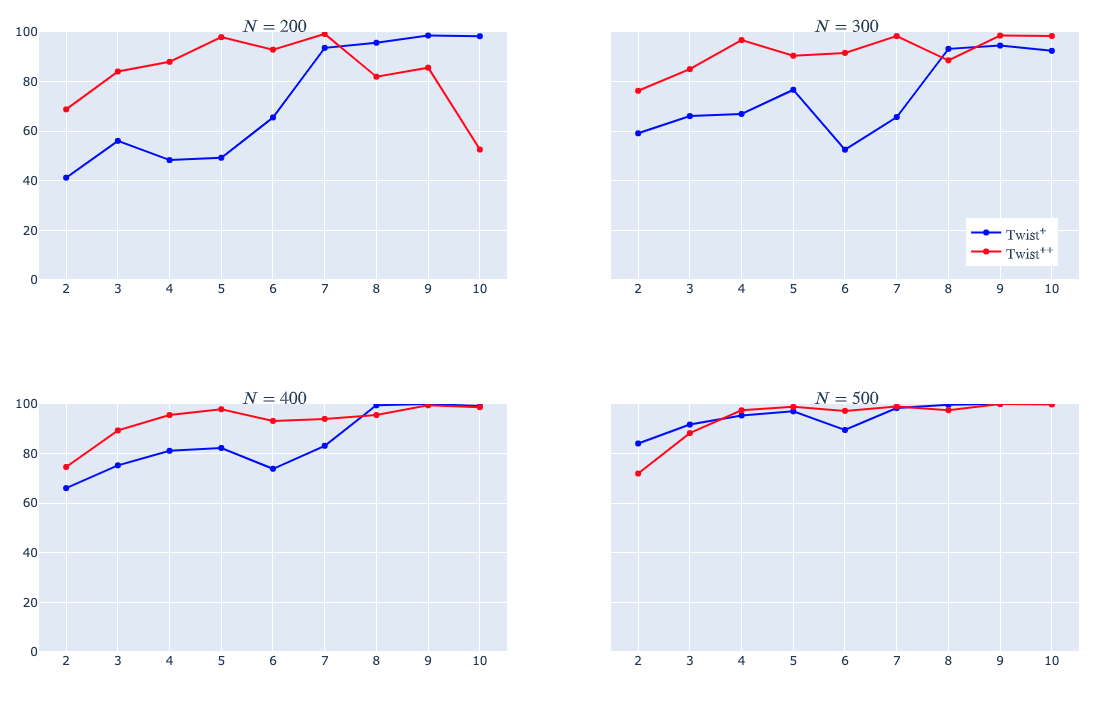}
    \caption{A comparison of the twist$^+$ algorithm and the  twist$^{++}$ algorithm with respect to $ 2 \le m \le 15$. The $x$-axis represents possible key lengths and the $y$-axis represents the success rate in percentage. Each graph is relative to a fixed message length $N$ designated at the top of the graph.}
    \label{fig:twistdoublecomparison}
\end{figure}

\begin{figure}[H]
    \centering
    \includegraphics[scale=0.33]{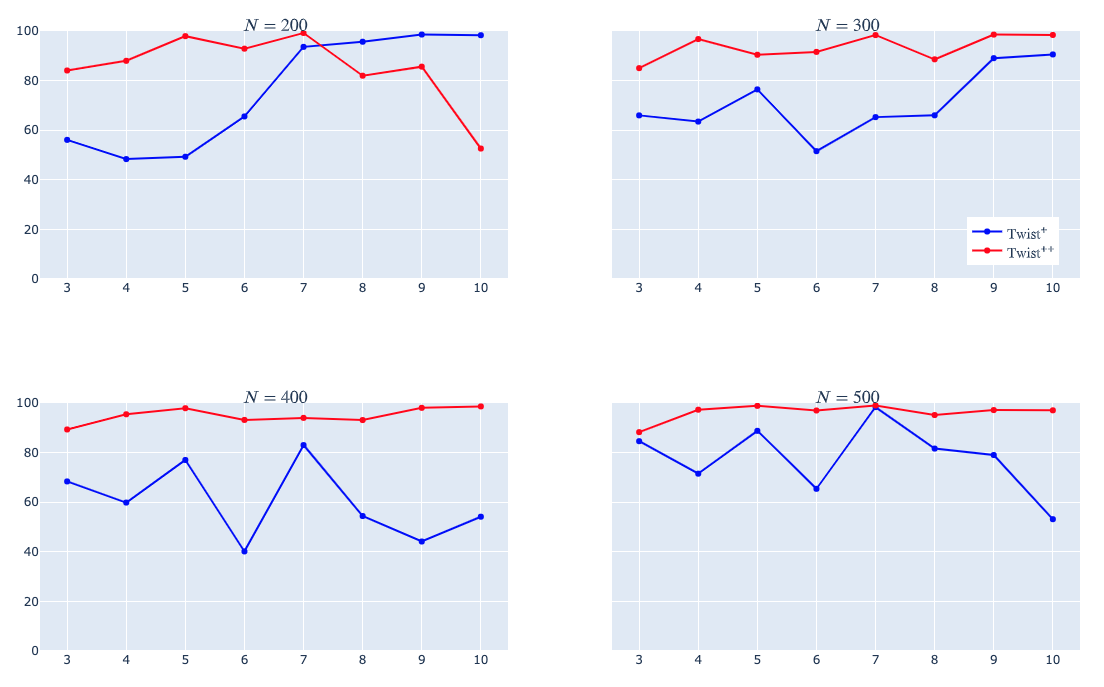}
    \caption{A comparison of the twist$^+$ algorithm and the  twist$^{++}$ algorithm with respect to $ 2 \le m \le 20$. The $x$-axis represents possible key lengths and the $y$-axis represents the success rate in percentage. Each graph is relative to a fixed message length $N$ designated at the top of the graph.}
    \label{fig:twistdoublecomparison_m20}
\end{figure}

\begin{table}[h!]
\begin{center}
\label{Table:Twistplusplusfails}
\begin{tabular}{ |c|c|c|c|c|c| } \hline
 Number of texts & $N$  & $k$  & Incorrect predictions & Predicted largest divisor & Percentage  \\ \hline
 4000 & 200 & 8 & 721 & 661 & $91.68\%$ \\ \hline
  4000 & 200 & 9 & 575 & 345 & $60\%$\\  \hline
4000 & 200 & 10 & 1896 & 1313 & $69.25\%$ \\ \hline
\end{tabular}
\caption{Under parameters where the twist$^{++}$ algorithm performs weakest, it usually predicts the largest (nontrivial) divisor of the actual key length. Here, $k$ stands for actual key length, $N$ stands for text length, and the percentage column represents the percentage of incorrect predictions that were largest (nontrivial) divisor of the actual key length.}
\end{center}
\end{table}

We now provide two examples comparing the twist-based algorithms to help give a deeper understanding of their differences. Our first example illustrates what can happen when the twist$^{+}$ algorithm fails and predicts a multiple of the key length, while the twist$^{++}$ algorithm succeeds in predicting the actual key length. We then consider an example   where the twist$^{+}$ algorithm succeeds, but the twist$^{++}$ algorithm fails by predicting the largest nontrivial divisor of the key length.

\begin{ex}
\label{TwistEx1}
For our first example, we will be using part of the prologue of \textit{Romeo and Juliet} as our plaintext:

\begin{flushleft}
``Two households, both alike in dignity, \\
(In fair Verona, where we lay our scene), \\
From ancient grudge break to new mutiny, \\
Where civil blood makes civil hands unclean. \\
From forth the fatal loins of these two foes \\
A pair of star-cross'd lovers take their life; \\
Whose misadventur'd piteous overthrows \\
Doth with their death bury their parents’ strife. \\
The fearful passage of their death-mark'd love \\
And the continuance of their parents’ rage''
\end{flushleft}

Once punctuation and spacing are removed, the resulting text has $N=350$ letters. We then encrypted this message with the keyword ‘will’, i.e., key length $k = 4$. This gave us the following ciphertext $\mathcal{M}$:

\begin{flushleft}
PEZSK CDPDW WOOJZ EDIWT GMTYZ QRYEB JTJNL TNDPC KVLHD MCPSM WLUWF COKPY ANCZI IYNEM YECZF OCMMC AIVEK VPHIC ETJGH SAZPN EDTWX TZZZU LVAAN TRQWS WVODQ VNWAI YQNWX QKZES PPPQW BLWHW TYOWQ EDMDP PEZQK MDLLI TCKND EWZNC KADOH WGPNA ELGME SAQCW ENPHD WDPIQ DLZDP YPCCO LQEPK CDZRM CEDZZ HOLZE DETED BSPEZ OPWBS MQZJE DMTCL ICPJB DDPZT QABSP BMLCB CWAWA DLCMZ QPPPT NLPLP PXLNS OWKDPL JLESA KZYPQ YFWVN PKNES AQCAW ZPYPA CLCM
\end{flushleft}

We then analyzed how well different key length attacks performed for this ciphertext.  First, $IC(\mathcal{M}) \approx 0.0503$, which results in the IC predicting a key length of $m \approx 2.265$. The Babbage--Kasiski test found  distances of 16, 32, 56, 87, 104, 124, 128, 140, 156, 160, 220, 224, and 247 between repeated trigrams. The greatest common divisor of a majority of these distances is 16. Thus, this test predicts a key length that is a factor of 16. For the twist-based algorithms, we considered $m$-values  of $1 \leq m \leq 25$  to make sure we included indices (potential key lengths to check) that are the first couple of multiples of the actual key length $k=4$. For this example, the twist$^{+}$ algorithm predicts a key length of 16. Thus,  all of these tests predict a divisor of $k$, a multiple of $k$, or a set of possible key length values that includes $k$ along with factors and divisors of $k$. However, the twist$^{++}$ algorithm predicts a key length  of $4$ for this example, making this test the only one that correctly predicts the key length as the only possible value. 

We provide Table \ref{Table3.2} and Figure \ref{fig:TwistEx1} to help understand how these twist-based algorithms behave for this example. For these visuals, we only consider indices $1 \leq m \leq 17$ for formatting purposes. For $m \geq 18$, $T(\mathcal{M},m)$ quickly converges to $100$. This is expected since once $m$ is sufficiently large, each coset contains very few letters; see Remark \ref{remark1}. In fact, for any $m \geq 23$, $T(\mathcal{M},m) =100$ for this text. As a result, we  expect twist$^{++}$ indices to become sufficiently small for large $m$-values, making it extremely unlikely that the twist$^{++}$ algorithm will accidentally predict a large multiple of the actual key length, which was a significant issue for the twist algorithm by Lemma  \ref{lem:lambda m bigger than m}.  Here, the local maxima that occurs at $m= k (=4)$ for $T(\mathcal{M}, m)$ is also the point where $T(\mathcal{M}, m)$ has the largest left-hand and right-hand derivatives, i.e., when $T^{++}(\mathcal{M},m)$ index is maximized, which one can visually verify in Figure \ref{fig:TwistEx1}. Based on experimental evidence, this phenomenon frequently occurs for short key lengths, which  motivated the definition of the twist$^{++}$ index. While it is more difficult to visualize the behavior of the twist$^{+}$ algorithm based on the graph of the twist indices, we do see that the graph of $T^{+}(\mathcal{M},m)$ has local maxima at $m = \lambda k$ with $\lambda \in \{1,2,3,4\}$  and a global maximum at $m = 16 ( =4k)$, for the given domain. Notice that for $T^{+}(\mathcal{M},m)$, the third highest local minima occurs at $m=4$ (the actual key length value) and the second highest local minima occurs at $m=12$. Thus, as we vary our domain of $m$-values, the twist$^{+}$ algorithm could make a variety of predictions about key lengths and would only predict correctly if $m \leq 11$ are considered. This emphasizes the importance of clearly stating your domain of $m$-values for twist-based algorithms.

\begin{table}
\begin{center}
\label{Table3.2}
\begin{tabular}{ |c|c|c|c|c|c|c|c|c|c|c|c|c|c|c|c|c|c| } 
 \hline
 $m$ & 1 & 2 & 3 & 4 & 5 & 6 & 7 & 8 & 9 & 10 & 11 & 12 & 13 & 14 &15 & 16 & 17 \\ \hline
 $T(\mathcal{M},m)$ & 50 & 59 & 58  & 77 & 63 & 67 & 64 & 83 & 73 & 81 & 78 & 91 & 82 & 87 & 89 & 99 & 90 \\ 
 \hline
 $T^{+}(\mathcal{M},m)$ &  & 9 & 3 & 21 & 2 & 5 & 2 & 20 & 7 & 15 & 10 & 22 & 11 & 16 & 17 & 26 & 15 \\ \hline
 $T^{++}(\mathcal{M},m)$ &  & 5 & -11 & 17 & -9 & 3 & -11 & 15 & -9 & 6 & -8 & 11 & -7 & 1 & -4 & 10 & -8 \\ \hline
\end{tabular}
\caption{Twist indices for Example \ref{TwistEx1}. Values are rounded to the nearest integer.}
\end{center}
\end{table}


\begin{figure}
	\centering
	\begin{overpic}[scale = 0.90]{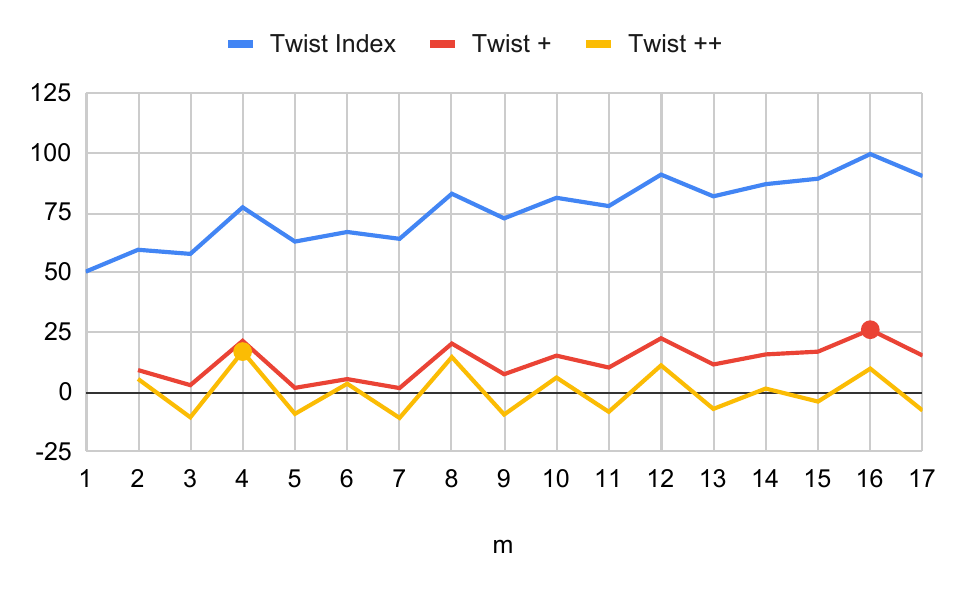}
	\end{overpic}
 \caption{A graphical depiction of the results for the twist-based algorithms for Example \ref{TwistEx1}. For $T^{+}(\mathcal{M},m)$ and $T^{++}(\mathcal{M},m)$, we have marked the absolute maximums for these graphs.}
	\label{fig:TwistEx1}
\end{figure}

\end{ex}

\begin{ex}
\label{TwistEx2}
For our second example we will be using part of the poem Overlooked by Emily Pauline Johnson as our plaintext: \begin{flushleft}
``Has passed me by.\\
I called, "O stay thy flight," but all unheard \\
                    My lonely cry: \\
    O! Love, my tired heart had need of thee! \\
    Is thy sweet kiss withheld alone from me? \\

    Sleep, sister-twin of Peace, my waking eyes \\
                    So weary grow! \\
    O! Love, thou wanderer from Paradise, \\
                    Dost t'' 
\end{flushleft}
Once punctuation and spacing are removed, the resulting text has $N=200$ letters. We encrypted this message with the keyword ``firebird'', i.e., key length $k=8$. This gave us the following ciphertext $\mathcal{M}$: 

\begin{flushleft}
MIJTB AJHIU VFZQT DQTVH PAKDD BYCGT ZJMBS YUICO ZVYIB ZUPDT FRFTP FWGFP PDVPD BZVFL YHFZK LBLEH JLFJU PVHNA KLZAN HJBBM TANLY PYIML ROTVV JSWDP JACIF XJLXB VVUEZ QTNGI BKVPD EROJV XHDMJ WPEVD WGXVP EFOTD VXIWL ZFVUI SMIIW WDTBZ RGNAV HPAKW
\end{flushleft}

For this text, $IC(\mathcal{M}) \approx 0.0420$ which results in the IC predicting a key length of $m \approx 6.7997$. The Babbage--Kasiski test found distances of $80$, $122$, and $176$. The $gcd(80, 122, 176) =2$, while the $gcd(80, 176) = 16$.  So the Babbage--Kasiski test directs one towards a divisor of $16$, though with minimal evidence. Finally, the twist$^{+}$ algorithm predicts a key length of $8$, while twist$^{++}$ algorithm predicts a key length of $4$ (the largest nontrivial divisor of $k=8$). As discussed earlier in this section, this behavior is expected for the twist$^{++}$ algorithm when it makes an incorrect prediction for such $N$ and $k$. 

We provide Table \ref{table:3} and Figure \ref{fig:TwistEx2} to assist with this example. Here, $m$-values larger than the ones shared only result in small variations of $T(\mathcal{M},m)$-values and do not maximize $T^{+}(\mathcal{M},m)$ nor $T^{++}(\mathcal{M},m)$. As discussed earlier, we can see that the twist$^{++}$ algorithm is very unlikely to predict a multiple of $k$ since there is so little variation in $T(\mathcal{M}, m)$ when $m>k$ under such parameters. We can also see that the second largest $T^{++}(\mathcal{M},m)$-value occurs at $m=8$, the actual key length. 

\begin{table}
\begin{center}

\begin{tabular}{ |c|c|c|c|c|c|c|c|c|} 
 \hline
 $m$ & 1 & 2 & 3 & 4 & 5 & 6 & 7 & 8  \\ \hline
 $T(\mathcal{M},m)$ & 40	& 57 &	52.0277 &	78 &	58 &	71.925 &	65.975 &	90 	  \\ 
 \hline
 $T^{+}(\mathcal{M},m)$ &  & 17	& 3.528 &	28.324	& 1.243 &	14.920	& 6.483 &	29.582 	
 \\ \hline
 $T^{++}(\mathcal{M},m)$ &  & 10.986	& -15.472 &	22.986 &	-16.963 &	9.938 &	-14.988 &	15.006 
 \\ \hline
 $m$ & 9 & 10 & 11 & 12 & 13 & 14 & 15 & 16  \\ \hline 
$T(\mathcal{M},m)$ & 84.014 &	88	& 93.993 &	100 &	99.0385  &	100 &	99.048 &	100	 \\ \hline
$T^{+}(\mathcal{M},m)$ & 19.898 &	21.673 &	25.498 &	29.188 &	25.794 &	24.771 &	22.050  &	21.532	 \\ \hline
$T^{++}(\mathcal{M},m)$ &	-4.986 &	-1.003 &	-0.007 &	3.484 &	-0.962 &	0.957 &	-0.952 &	0.476	\\ \hline
\end{tabular}
\caption{Twist indices for Example \ref{TwistEx2}. Values are rounded to three decimal places.}
\label{table:3}
\end{center}

\end{table}

\begin{figure}
	\centering
	\begin{overpic}[scale = 0.90]{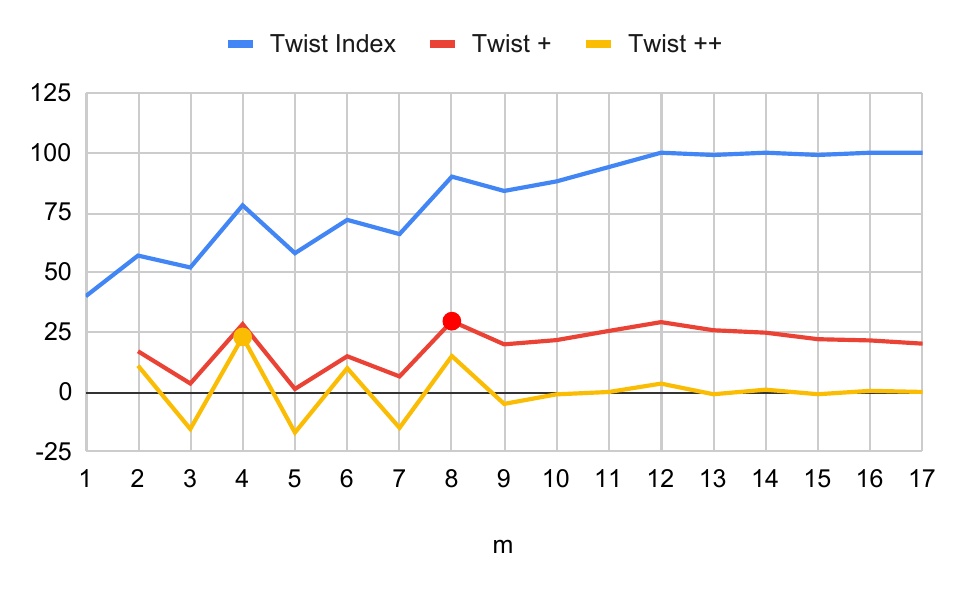}
	\end{overpic}
 \caption{A graphical depiction of the results for the Twist-based algorithms for Example \ref{TwistEx2}. For $T^{+}(\mathcal{M},m)$ and $T^{++}(\mathcal{M},m)$, we have marked the absolute maximums for these graphs.}
	\label{fig:TwistEx2}
\end{figure}
\end{ex}


\section{Funding Details}
This work was financially supported by the Furman University Department of Mathematics  via the Summer Mathematics Undergraduate Research Fellowships.

\section{Biographical Note}

\noindent \textbf{Christian Millichap} is an Associate Professor of Mathematics at Furman University in Greenville, SC. His research interests are in geometric topology and knot theory. He has also enjoyed teaching a variety of classes in cryptology for high school students and undergraduates. 
\\\\
\textbf{Yeeka Yau} is an Assistant Professor of Mathematics at the University of North Carolina Asheville. His research interests are in Coxeter groups, combinatorial and geometric group theory, cryptology and machine learning.
\\\\
\textbf{Alyssa Pate}  is a student at Furman University. She is studying Applied Mathematics and Data Analytics with hopes of going into the Data Analytics field. With friend and fellow collaborator, Morgan Carns, she participated in the 2023 Kryptos Competition. She enjoys learning more about Cryptology and hopes to continue to discover more in the future.
\\\\
\textbf{Morgan Carns}  is a student at Furman University getting her Bachelor of Science in Applied Mathematics with a Data Analytics Minor.  Her cryptology experience started with the Kryptos competition held by Central Washington University where she placed as an amateur codebreaker for solving one of three challenges.  She then continued her exploration in cryptology by completing summer research through the Furman Mathematics Department in the summer of 2023.  She hopes to continue her cryptology experience throughout the rest of her life.

\section{Data Availability Statement}

The data that support the findings of this study are openly available in Zenodo at \\ \href{https://doi.org/10.5281/zenodo.8373071}{https://doi.org/10.5281/zenodo.8373071}.

\bibliographystyle{hamsplain}
\bibliography{biblio}

\end{document}